%% file: root.tex
\documentclass[journal,twoside,web]{ieeecolor}
\usepackage{lcsys}
\pdfoutput=1

\makeatletter
\let\NAT@parse\undefined
\makeatother

\usepackage{amsmath,amssymb}
\usepackage{mathtools}
\usepackage{cite}
\usepackage{color}
\usepackage{caption}
\newenvironment{myproof}[1] {\noindent\hspace{2em}{\itshape {#1}: }}{\hspace*{\fill}~\QED\par\endtrivlist\unskip}

\newtheorem{assumption}{Assumption}
\newtheorem{definition}{Definition}
\newtheorem{theorem}{Theorem}
\newtheorem{proposition}{Proposition}
\newtheorem{lemma}{Lemma}
\newtheorem{claim}{Claim}

\newtheorem{remark}{Remark}

\usepackage{fancyhdr}
\usepackage{xcolor}

\usepackage{hyperref}
\hypersetup{
	colorlinks,
	linkcolor={red!50!black},
	citecolor={blue!50!black},
	urlcolor={blue!80!black}
}

\def\doi{10.1109/LCSYS.2023.3286174}
\fancyfootoffset{-1em}
\fancypagestyle{copyright}{
	\cfoot{\vspace{-1.8em}\footnotesize\textcopyright{} 2023 IEEE. Personal use of this material is permitted.  Permission from IEEE must be obtained for all other uses, in any current or future media, including reprinting/republishing this material for advertising or promotional purposes, creating new collective works, for resale or redistribution to servers or lists, or reuse of any copyrighted component of this work in other works.}
	\chead{\footnotesize This article has been accepted for publication in IEEE Control Systems Letters. This is the author's version which has not been fully edited and content may change prior to final publication. Citation information: \href{http://dx.doi.org/\doi}{DOI \doi}}
}

\begin{document}

\title{On an integral variant of incremental input/output-to-state stability and its use as a notion of nonlinear detectability}

\author{
	Julian D. Schiller and Matthias A. Müller
	\thanks{This work was funded by the Deutsche Forschungsgemeinschaft (DFG, German Research Foundation) --  426459964. \textit{(Corresponding author: Julian D. Schiller.)}}
	\thanks{The authors are with the Leibniz University Hannover, Institute of Automatic Control, 30167 Hannover, Germany (e-mail: schiller@irt.uni-hannover.de; mueller@irt.uni-hannover.de).}
}

\maketitle
\thispagestyle{empty}
\thispagestyle{copyright}

\begin{abstract}
	We propose a time-discounted integral variant of incremental input/output-to-state stability (i-iIOSS) together with an equivalent Lyapunov function characterization.
	Continuity of the i-iIOSS Lyapunov function is ensured if the system satisfies a certain continuity assumption involving the Osgood condition.
	We show that the proposed i-iIOSS notion is a necessary condition for the existence of a robustly globally asymptotically stable observer mapping in a time-discounted ``$\boldsymbol{L^2}$-to-$\boldsymbol{L^{\infty}}$'' sense.
	In combination, our results provide a general framework for a Lyapunov-based robust stability analysis of observers for continuous-time systems, which in particular is crucial for the use of optimization-based state estimators (such as moving horizon estimation).	
\end{abstract}

\begin{IEEEkeywords}
	Incremental system properties, nonlinear systems, stability, detectability, state estimation.
\end{IEEEkeywords}

\section{Introduction}
\label{sec:intro}

\IEEEPARstart{T}{he} concept of incremental IOSS (i-IOSS) was originally proposed in \cite{Sontag1997} as a more appropriate notion of nonlinear detectability than IOSS (which can merely be viewed as ``zero-detectability'').
Introduced in an ``$L^\infty$-to-$L^\infty$'' sense (where we adopt the terminology from~\cite{Sontag1998}), it has been shown that a continuous-time system must necessarily satisfy the i-IOSS property to admit a robustly stable observer, and its discrete-time analogue has become the standard in the field of optimization-based state estimation, cf.~\cite{Knuefer2023, Allan2021a, Schiller2023c}.

The characterization of system properties via Lyapunov functions linked to converse theorems that establish necessity of their existence has turned out to be very useful for system analysis and the design of controllers and observers.
Such results are available for, e.g., global asymptotic stability (GAS) in~\cite{Lin1996} and (non-incremental) IOSS in~\cite{Krichman2001}, albeit under the condition that external signals (e.g., time-varying parameters, disturbances) of the system take values in a compact set.
Converse Lyapunov theorems for the incremental ``$L^2$-to-$L^\infty$'' versions of GAS and input-to-state stability (ISS) were considered in~\cite{Angeli2002} and~\cite{Angeli2009}, respectively, where the condition of compactness could be weakened by using a dissipation inequality in integral form along with relaxing the requirement of smoothness of the Lyapunov function to mere continuity.

Time-discounted variants of i-IOSS were proposed in \cite{Knuefer2020,Allan2021} for discrete-time systems, where it was shown that discounting past disturbances appears very natural and even without loss of generality.
A corresponding converse Lyapunov result is provided in \cite{Allan2021}, which is structurally easier and more intuitive to establish with such a discount factor than without, as is the case in, e.g., \cite{Lin1996,Krichman2001,Angeli2002,Angeli2009}.
Moreover, \mbox{i-IOSS} with time-discounting and its associated Lyapunov function are crucial for recent results in the field of optimization-based state estimation for discrete-time systems, cf., e.g., \cite{Knuefer2023,Allan2021a,Schiller2023c}.

Inspired by these works, we propose a discrete-time integral (``$L^2$-to-$L^\infty$'') variant of i-IOSS for continuous-time systems, namely i-iIOSS (Section~\ref{sec:setup}).
As the main contribution, we show that i-iIOSS is equivalent to the existence of a continuous i-iIOSS Lyapunov function (Section~\ref{sec:conv}), where an exponential decay can be used without loss of generality (cf. Remark~\ref{rem:exp}).
Our proofs use similar tools as in previous works on incremental integral ISS \cite{Angeli2009} and i-IOSS in the discrete-time setting \cite{Allan2021}; however, we point out that the presented results do not straightforwardly follow from them.
In particular, continuity of the Lyapunov function candidate is shown by replacing the standard local Lipschitz assumption on the dynamics $f$ by a global property involving the Osgood condition~\cite{Osgood1898}, which is quite novel in the context of control systems.
As a byproduct, based on this assumption, we formally prove global existence and uniqueness of system trajectories by adapting the results from \cite{Lipovan2000,Bihari1956} to the generic class of inputs considered here.

As second contribution, we propose a time-discounted integral ``$L^2$-to-$L^\infty$" variant of robust global asymptotic stability and show necessity of i-iIOSS for a system to admit a general observer mapping satisfying this property (Section~\ref{sec:detect}).
Asking such a stability property from an observer is advantageous for several reasons: first, it can be seen as accounting for the disturbance energy under fading memory and thus allows for a physical interpretation; second, it directly implies an $L^\infty$ error bound and thus combines the advantages of classical and integral ISS properties.
In combination, we provide a general framework for a Lyapunov-based robust stability analysis of observers in continuous-time. This is an essential tool in the context of moving horizon estimation in~\cite{Schiller2023b}, where we also provide constructive conditions for \mbox{i-iIOSS} Lyapunov functions that can be efficiently verified~in~practice.

\subsubsection*{Notation}
Let $\mathbb{R}_{\geq0}$ denote the set of all non-negative real values.
By $|x|$, we indicate the Euclidean norm of the vector $x\in\mathbb{R}^n$. For a measurable, essentially bounded function $z:\mathbb{R}_{\geq0}\rightarrow\mathbb{R}^n$, the essential sup-norm is defined as $\|z\| = \mathrm{ess}\sup_{t\geq0}|z(t)|$, and for the restriction of $z$ to an interval $[a,b]$ with $a,b\geq0$ by $\|z\|_{a:b} = \mathrm{ess}\sup_{t\in[a,b]}|z(t)|$.
We recall that a function $\alpha:\mathbb{R}_{\geq 0}\rightarrow\mathbb{R}_{\geq 0}$ is of class $\mathcal{K}$ if it is continuous, strictly increasing, and satisfies $\alpha(0)=0$; if additionally $\alpha(s)=\infty$ for $s\rightarrow\infty$, it is of class $\mathcal{K}_{\infty}$.

\section{Setup and Preliminaries}
\label{sec:setup}

We consider continuous-time systems of the form%
\begin{subequations}\label{eq:sys}
	\begin{align}
	\dot{x}(t) &= f(x(t),u(t),d(t)),\label{eq:sys_1}\\
	y(t) &= h(x(t),u(t),d(t))\label{eq:sys_2}
	\end{align}
\end{subequations}
with states $x\in\mathcal{X}\subseteq\mathbb{R}^n$ ($0\in\mathcal{X}$), outputs $y\in\mathcal{Y}\subseteq\mathbb{R}^p$ ($0\in\mathcal{Y}$), and time $t\geq0$.
The input $u$ and the time-varying parameter $d$ are measurable, locally essentially bounded functions taking values in $\mathcal{U}\subseteq\mathbb{R}^m$ and $\mathcal{D}\subseteq\mathbb{R}^q$ with $0\in\mathcal{U},\mathcal{D}$, and we denote the set of such functions as $\mathcal{M}_{\mathcal{U}}$ and $\mathcal{M}_{\mathcal{D}}$, respectively.
The solution of \eqref{eq:sys_1} at $t\geq0$ with initial state $\chi\in\mathcal{X}$ and input signals $u\in\mathcal{M}_{\mathcal{U}}$ and $d\in\mathcal{M}_{\mathcal{D}}$ is denoted by $x(t,\chi,u,d)$, and the output signal by $y(t,\chi,u,d) := h(x(t,\chi,u,d),u(t),d(t))$.

\begin{assumption}\label{ass:f}
	The function $f:\mathcal{X}\times\mathcal{U}\times{\mathcal{D}}\rightarrow \mathcal{X}$ satisfies $f(0,0,0)=0$ and
	\begin{align}
		&\ |f(x_1,u_1,d_1)-f(x_2,u_2,d_2)|\nonumber\\
		\leq &\ \kappa_1\left(|(x_1,u_1,d_1)-(x_2,u_2,d_2)|\right)\label{eq:ass_f_k}
	\end{align}
	for all $x_1,x_2\in\mathcal{X}$, all $u_1,u_2\in\mathcal{U}$, and all $d_1,d_2\in\mathcal{D}$, where $\kappa_1: \mathbb{R}_{\geq0} \rightarrow \mathbb{R}_{\geq0}$ is continuous, non-decreasing, $\kappa_1(0)=0$, $\kappa_1(s)>0$ for all $s>0$, and
	\begin{equation}\label{eq:ass_f}
		\int_{0}^{1}\frac{ds}{\kappa_1(3s)} = \infty, \quad \int_{1}^{\infty}\frac{ds}{\kappa_1(3s)} = \infty.\\[1.5ex]
	\end{equation}
\end{assumption}

\begin{assumption}\label{ass:h}
	The function $h$ satisfies
	\begin{equation}
	|h(x_1,u_1,d)-h(x_2,u_2,d)|\leq \kappa_2(|(x_1,u_1)-(x_2,u_2)|)
	\end{equation}
	for some $\kappa_2\in\mathcal{K}_{\infty}$, for all $x_1,x_2\in\mathcal{X}$ and all $u_1,u_2\in\mathcal{U}$ uniformly in $d\in\mathcal{D}$.
\end{assumption}

Assumption~\ref{ass:f} is essential for proving the converse Lyapunov theorem below (the factor $3$ in~\eqref{eq:ass_f} is required for technical reasons).
It replaces the usual assumption of $f$ being locally Lipschitz (which is not suitable in our case, cf.~Remark~\ref{rem:local_L} in the appendix) and ensures global existence and uniqueness of solutions of~\eqref{eq:sys}.
Inequality \eqref{eq:ass_f_k} together with the first equation in~\eqref{eq:ass_f} is similar to the so-called Osgood condition, which was originally proposed in \cite{Osgood1898} to establish local uniqueness of solutions without employing a Lipschitz property.
The second equation in \eqref{eq:ass_f} ensures that these solutions exist globally in time.
(A similar condition is, in fact, necessary for the global existence of solutions to the scalar differential equation $\dot{v}=\kappa_1(v)$, cf.~\cite{Constantin1995}.)
Valid functions that satisfy Assumption~\ref{ass:f} are, e.g., $s \mapsto s$, $s \mapsto s\ln(s+1)$, cf. also \cite{Constantin1995,Lipovan2000,Osgood1898}.
While especially the second condition in~\eqref{eq:ass_f} may be a limitation, we point out that global existence of solutions is often assumed in the literature and significantly facilitates the exposition.

The main properties are summarized in the following proposition; the proof is shifted to the appendix (cf. Section~\ref{sec:app_f}) and requires a straightforward extension of the results from \cite{Lipovan2000,Bihari1956} by addressing the generic class of inputs considered here.

\begin{proposition}\label{prop:f}
	Let Assumption 1 hold. Then, \eqref{eq:sys_1} admits a unique solution defined on $\mathbb{R}_{\geq0}$ for all $\chi\in\mathcal{X}$, all $u\in\mathcal{M}_{\mathcal{U}}$, and all $d\in\mathcal{M}_{\mathcal{D}}$.
\end{proposition}

The topic of this paper is a notion of nonlinear detectability in terms of the following notion.

\begin{definition}[\protect{i-iIOSS}]
	\label{def:dIOSS}
	The system~\eqref{eq:sys} is incrementally integral input/output-to-state stable (i-iIOSS) if there exist some $\alpha,\alpha_{x},\alpha_{u},\alpha_{y}\in\mathcal{K}_{\infty}$, and $\lambda\in[0,1)$ such that
	\begin{align}	
		\label{eq:dIOSS}
		\alpha(|x_\Delta(t)|) &\leq
		\alpha_{x}(|\chi_\Delta|)\lambda^{t} \nonumber \\
		&\hspace{1.3ex} + \int_0^t  \lambda^{t-\tau}\big(\alpha_u(|u_\Delta(\tau)|) +  \alpha_{y}(|y_\Delta(\tau)|)\big)d\tau
	\end{align}
	for all $t\geq0$, all $\chi_1,\chi_2\in\mathcal{X}$, all $u_1,{u}_2\in\mathcal{M}_{\mathcal{U}}$, and all $d\in\mathcal{M}_{\mathcal{D}}$, where $z_\Delta = z_1-z_2$ for $z=\{\chi,x,u,y\}$ with $x_i(t) = x(t,\chi_i,u_i,d)$ and $y_i(t) = y(t,\chi_i,u_i,d)$, $i=1,2$.
\end{definition}

We additionally propose the following equivalent (cf.~Theorem~\ref{thm:converse} below) Lyapunov function characterization.

\begin{definition}[i-iIOSS Lyapunov function]
	\label{def:IOSS_Lyap}
	A function $U:\mathcal{X}\times\mathcal{X}\rightarrow\mathbb{R}_{\geq 0}$ is an i-iIOSS Lyapunov function if it is continuous and there exist functions $\alpha_1,\alpha_2,\sigma_w,\sigma_{y}\in\mathcal{K}_{\infty}$ and a constant $\lambda\in[0,1)$ such that
	\begin{subequations}
		\label{eq:IOSS_Lyap}
		\begin{align}
		\label{eq:IOSS_Lyap_1}
		&\ \alpha_1(|\chi_\Delta|)\leq U(\chi_1,\chi_2) \leq \alpha_2(|\chi_\Delta|),\\	
		&\ U(x_1(t),x_2(t)) \label{eq:IOSS_Lyap_2}\\
		\leq&\ U(\chi_1,\chi_2)\lambda^{t} + \int_{0}^{t}\lambda^{t-\tau}\big(\sigma_{u}(|u_\Delta(\tau)|) +\sigma_{y}(|y_{\Delta}(\tau)|)\big)d\tau \nonumber
		\end{align}
	\end{subequations}
	for all $t\geq0$, all $\chi_1,\chi_2\in\mathcal{X}$, all $u_1,{u}_2\in\mathcal{M}_{\mathcal{U}}$, and all $d\in\mathcal{M}_{\mathcal{D}}$, where $z_\Delta = z_1-z_2$ for $z=\{\chi,x,u,y\}$ with $x_i(t) = x(t,\chi_i,u_i,d)$ and $y_i(t) = y(t,\chi_i,u_i,d)$, $i=1,2$.
\end{definition}

The integral form of \eqref{eq:dIOSS} and \eqref{eq:IOSS_Lyap_2} together with the continuity of $U$ is motivated by \cite{Angeli2002,Angeli2009}, originally employed to allow for a non-compact input set $\mathcal{D}$, where smooth converse Lyapunov theorems usually fail, cf. \cite[Rem.~2.4]{Angeli2002}, \cite[Sec.~8]{Lin1996}.
The exponential decrease in~\eqref{eq:dIOSS} and \eqref{eq:IOSS_Lyap_2} is motivated by recent results in the discrete-time literature, where this is crucial to develop full information and moving horizon estimation schemes with suitable stability properties, cf.~\cite{Schiller2023c,Allan2021a}.
In~\cite{Schiller2023b}, we showed that this carries over to the continuous-time setting, where we also provide sufficient conditions for the construction of i-iIOSS Lyapunov functions for special systems classes.

\begin{remark}\label{rem:exp}
	Considering an exponential decrease in Lyapunov coordinates is without loss of generality; in fact, one can straightforwardly show that~\eqref{eq:IOSS_Lyap_2} is equivalent to a dissipation inequality in the form of $U(x_1(t),x_2(t)) - U(\chi_1,\chi_2) \leq \int_{0}^{t}(-\alpha_3(|x_\Delta(\tau)|) + {\sigma}_{u}(|u_{\Delta}(\tau)|)+\sigma_y(|y_\Delta(\tau)|))ds$ with $\alpha_3\in\mathcal{K}_{\infty}$, which follows by application of the standard comparison lemma and~\cite[Prop.~13]{Praly1996}, similar to, e.g.,~\cite[Lem.~10]{Sontag1997}.
\end{remark}

\begin{remark}\label{rem:integral}
	A key advantage of the discrete-time i-IOSS counterpart is that discounted summation and discounted maximization are in some sense equivalent, cf.~\cite{Knuefer2020,Allan2021}.
	This does not carry over to the continuous-time setting (the discounted integral in~\eqref{eq:dIOSS} could indeed be transferred to a discounted $L^{\infty}$-norm bound, but not \textit{vice versa}, unless strong regularity assumptions on $u$ are enforced).
	As a result, the proposed notion from Definition~\ref{def:dIOSS} implies i-IOSS in an ``$L^{\infty}$-to-$L^{\infty}$'' sense with time-discounting (similar to \cite[Def.~2.4]{Allan2021}) and without discounting \cite[Def.~22]{Sontag1997}, cf.~\cite[Prop.~1]{Schiller2023b}.
	Investigating some converse implications may be subject of future research.
\end{remark}

\section{Converse Lyapunov theorem for i-iIOSS}
\label{sec:conv}

We now show equivalence between the proposed i-iIOSS characterizations from Definitions~\ref{def:dIOSS} and~\ref{def:IOSS_Lyap}.

\begin{theorem}\label{thm:converse}
	Let Assumptions~\ref{ass:f} and~\ref{ass:h} hold. The system~\eqref{eq:sys} is i-iIOSS if and only if there exists an i-iIOSS Lyapunov function.
\end{theorem}

\begin{proof}
	\textit{Part I (Sufficiency):}
	The implication follows by applying the bounds~\eqref{eq:IOSS_Lyap_1} to~\eqref{eq:IOSS_Lyap_2}, which directly yields~\eqref{eq:dIOSS}.
	
	\textit{Part II (Necessity):}
	The proof uses and combines similar arguments as in previous converse theorems and Lyapunov function constructions, in particular \cite{Allan2021,Ingalls2001}; continuity of the Lyapunov function is proven in a different fashion invoking Assumptions~\ref{ass:f} and~\ref{ass:h}, cf.~Claim~\ref{prop:c} below and Appendix~\ref{sec:app_c}.

	For arbitrary $\chi_1,\chi_2\in\mathcal{X}$, we consider the following i-iIOSS Lyapunov function candidate\footnote{
		The inputs $u_1,u_2,d$ in \eqref{eq:def_lyap} are maximized over the sets $\mathcal{M}_{\mathcal{U}}$ and $\mathcal{M}_{\mathcal{D}}$, respectively, which is omitted throughout this paper for brevity.
	}
	\begin{align}
	&\ U(\chi_1,\chi_2) \nonumber\\
	:=&\ \sup_{t\geq 0,u_1,u_2,{d}}
	\lambda^{-t/2}\Big(\alpha(|x(t,\chi_1,u_1,d)-x(t,\chi_2,u_2,d)|)\nonumber\\
	&\hspace{1ex} - \int_0^\infty\lambda^{t-\tau}2\alpha_u(|u_1(\tau)-u_2(\tau)|)d\tau\nonumber\\
	&\hspace{1ex} - \int_0^t \lambda^{t-\tau}\alpha_y(|y(t,\chi_1,u_1,d)-y(t,\chi_2,u_2,d)|)d\tau \Big)\label{eq:def_lyap}
	\end{align}
	and start by establishing the bounds~\eqref{eq:IOSS_Lyap_1}.
	For the term $\alpha(|x(t,\chi_1,u_1,d)-x(t,{\chi}_2,{u}_2,d)|)$ in \eqref{eq:def_lyap}, we can directly use the estimate from i-iIOSS~\eqref{eq:dIOSS}, which yields
	\begin{align*}
		U(\chi_1,\chi_2) &\leq \sup_{t\geq 0,u_1,u_2,d}
		\alpha_x(|\chi_1-\chi_2|)\lambda^{t/2} = \alpha_x(|\chi_1-\chi_2|),
	\end{align*}
	i.e., the upper bound in~\eqref{eq:IOSS_Lyap_1} with $\alpha_2=\alpha_x$.
	The lower bound follows by considering the candidate inputs $u_1=u_2$ and $t=0$, leading to $\alpha_1=\alpha$ in~\eqref{eq:IOSS_Lyap_1}.
	
	The following claim is proven in Appendix B.
	\begin{claim}\label{prop:c}
		The function $U$ in \eqref{eq:def_lyap} is continuous on $\mathcal{X}\times\mathcal{X}$.
	\end{claim}
	It remains to establish the dissipation inequality~\eqref{eq:IOSS_Lyap_2}. To this end, consider  $\zeta_1,\zeta_2\in\mathcal{X}$, $u_1,u_2\in\mathcal{M}_{\mathcal{U}},d\in\mathcal{M}_{\mathcal{D}}$, yielding the trajectories
	$z_j(t):=x(t,\zeta_j,u_j,d)$ with $j=1,2$ for $t\geq0$.
 	We obtain
	\begin{align}
		&\ U(z_1(t),z_2(t)) \label{eq:proof_conv_1}\\
		=&\ \sup_{\bar{t}\geq 0,\bar{u}_1,\bar{u}_2,\bar{d}}
		\lambda^{-\bar{t}/2}\Big(\alpha(|x(\bar{t},z_1(t),\bar{u}_1,\bar{d})-x(\bar{t},z_2(t),\bar{u}_2,\bar{d})|)\nonumber\\
		& -\int_0^{\infty}\lambda^{\bar{t}-\tau}2\alpha_u(|\bar{u}_1(\tau)-\bar{u}_2(\tau)|) \nonumber\\
		& -\int_0^{\bar{t}}\lambda^{\bar{t}-\tau} \alpha_y(|y(\tau,z_1(t),\bar{u}_1,\bar{d})-y(\tau,z_2(t),\bar{u}_2,\bar{d})|)d\tau \Big).\nonumber
	\end{align}
	For two functions $\bar{u},u$ defined on $[0,\infty)$, let $\bar{u}\sharp_tu$ denote their concatenation at some fixed time $t\geq0$, i.e.,
	\begin{equation*}
		\bar{u}{\sharp_t}u(\tau) :=
		\begin{cases}
		u(\tau), &\tau\in[0,t]\\
		\bar{u}(\tau-t), &\tau\in(t,\infty).
		\end{cases}
	\end{equation*}
	Hence, in~\eqref{eq:proof_conv_1}, we can infer that
	\begin{align}
		&\ \alpha(|x(\bar{t},z_1(t),\bar{u}_1,\bar{d})-x(\bar{t},z_2(t),\bar{u}_2,\bar{d})|)\label{eq:proof_conv_2}\\
		=&\ \alpha(|x(\bar{t}+t,\zeta_1,\bar{u}_1{\sharp_t}u_1,\bar{d}{\sharp_t}d) -x(\bar{t}+t,\zeta_2,\bar{u}_2{\sharp_t}u_2,\bar{d}{\sharp_t}d)|).\nonumber
	\end{align}
	Similarly,
	\begin{align}
		&\ \int_0^{\infty}\lambda^{\bar{t}-\tau}2\alpha_u(|\bar{u}_1(\tau)-\bar{u}_2(\tau)|)d\tau\nonumber\\
		= &\ \int_0^{\infty}\lambda^{\bar{t}+t-\tau}2\alpha_u(|\bar{u}_1\sharp_tu_1(\tau)-\bar{u}_2\sharp_tu_2(\tau)|)d\tau\nonumber\\
		&\ - \int_0^{t}\lambda^{\bar{t}+t-\tau}2\alpha_u(|{u}_1(\tau)-{u}_2(\tau)|)d\tau
		\label{eq:proof_conv_3}
	\end{align}
	and
	\begin{align}
		&\ \int_0^{\bar{t}}\lambda^{\bar{t}-\tau}\alpha_y(|y(\tau,z_1(t),\bar{u}_1,\bar{d})-y(\tau,z_2(t),\bar{u}_2,\bar{d})|)d\tau\nonumber\\
		= &\ \int_0^{\bar{t}+t}\lambda^{\bar{t}+t-\tau}\alpha_y(|y(\bar{t}+t,\zeta_1,\bar{u}_1{\sharp_t}u_1,\bar{d}{\sharp_t}d)\nonumber\\
		&\hspace{17ex} -y(\bar{t}+t,\zeta_2,\bar{u}_2{\sharp_t}u_2,\bar{d}{\sharp_t}d)|)d\tau\nonumber\\
		& - \int_0^{t}\lambda^{\bar{t}+t-\tau}\alpha_y(|y(\tau,\zeta_1,u_1,d)-y(\tau,\zeta_2,u_2,d)|)d\tau. \label{eq:proof_conv_4}
	\end{align}
	Now define $\hat{t}:=\bar{t}+t$. Consequently, $U(z_1(t),z_2(t))$ in \eqref{eq:proof_conv_1} can be bounded using the substitutions from~\eqref{eq:proof_conv_2}--\eqref{eq:proof_conv_4} and the fact that $\lambda\leq\sqrt{\lambda}\in[0,1)$ as
	\begin{align*}
		&\ U(z_1(t),z_2(t)) \\
		\leq&\ \sup_{\hat{t}\geq 0,\hat{u}_1,\hat{u}_2,\hat{d}}
		\lambda^{(-\hat{t}+t)/2}\Big(\alpha(|x(\hat{t},\zeta_1,\hat{u}_1,\hat{d})-x(\hat{t},\zeta_2,\hat{u}_2,\hat{d})|)\nonumber\\
		&\hspace{6ex} - \int_0^{\infty}\lambda^{\hat{t}-\tau}2\alpha_u(|\hat{u}_1(\tau)-\hat{u}_2(\tau)|)d\tau \\
		&\hspace{6ex} - \int_0^{\hat{t}}\lambda^{\hat{t}-\tau} \alpha_y(|y(\hat{t},\zeta_1,\hat{u}_1,\hat{d}) - y(\hat{t},\zeta_2,\hat{u}_2,\hat{d})|) d\tau \Big)\\
		&\ + \int_0^{t}\lambda^{t-\tau}\big(2\alpha_u(|u_1(\tau)-u_2(\tau)|)\\
		&\ \hspace{12ex} + \alpha_y(|y(\tau,\zeta_1,u_1,d)-y(\tau,\zeta_2,u_2,d)|)\big)d\tau\\
		\leq&\ \sqrt{\lambda}^t U(\zeta_1,\zeta_2) + \int_0^{t}\sqrt{\lambda}^{t-\tau}\big(2\alpha_u(|u_1(\tau)-u_2(\tau)|)\\
		&\hspace{14ex} +\alpha_y(|y(\tau,\zeta_1,u_1,d)-y(\tau,\zeta_2,u_2,d)|)\big)d\tau,
	\end{align*}
	which establishes the dissipation inequality~\eqref{eq:IOSS_Lyap_2} by a suitable redefinition of $\lambda$ and hence concludes this proof.
\end{proof}

\section{Nonlinear detectability}
\label{sec:detect}

In this section, we establish necessity of i-iIOSS for the existence of an observer mapping satisfying an ISS-like robust stability property in a time-discounted ``$L^2$-to-$L^{\infty}$'' sense.
In this context, we let $u$ include all \emph{unknown} signals (such as process disturbances and measurement noise) and $d$ \emph{known} exogenous signals (such as control inputs).
Let the set $\mathcal{M}_{\mathcal{Y}}$ contain all measurable, locally essentially bounded functions defined on $[0,\infty)$ taking values in $\mathcal{Y}$.
For a function $z$ defined on $[0,\infty)$ and any $t\geq0$, we denote by $z_t$ the truncated signal given by $z_t(\tau) := z(\tau), \tau \in [0,t)$ and  $z_t(\tau):= 0,\tau\in [t,\infty)$.

\begin{definition}[State observer]\label{def:obs}	
	The mapping
	\begin{equation}
		P: \mathbb{R}_{\geq 0}\times\mathcal{X}\times\mathcal{M}_{\mathcal{U}}\times\mathcal{M}_{\mathcal{D}}\times\mathcal{M}_{\mathcal{Y}}\rightarrow{\mathcal{X}}\label{eq:obs_map}
	\end{equation}
	is a robustly globally asymptotically stable (RGAS) observer for the system~\eqref{eq:sys} if there exist functions $\beta,\beta_x,\beta_u,\beta_y\in\mathcal{K}_{\infty}$ and a constant $\eta\in[0,1)$ such that the estimate
	\begin{equation}\label{eq:obs}
		\hat{x}(t) = P(t,\bar{\chi},\bar{u}_t,d_t,\bar{y}_t), \ \hat{x}(0) = \bar{\chi}
	\end{equation}
	satisfies
	\begin{align}
		&\beta(|x(t)-\hat{x}(t)|) \leq \beta_x(|\chi-\bar{\chi}|)\eta^t \label{eq:RGAS} \\
		&+\int_0^t\eta^{t-\tau}\big(\beta_u(|u(\tau)-\bar{u}(\tau)|) + \beta_y(|y(\tau)-\bar{y}(\tau)|)\big)d\tau\nonumber
	\end{align}
 	for all $t\geq0$, all $\chi,\bar{\chi}\in\mathcal{X}$, $u,\bar{u}\in\mathcal{M}_{\mathcal{U}}$, $d\in\mathcal{M}_{\mathcal{D}}$, $\bar{y}\in\mathcal{M}_{\mathcal{Y}}$, where ${x}(\tau) = x(\tau,\chi,u,d)$ and ${y}(\tau) = y(\tau,\chi,u,d)$, $\tau\in[0,t]$.
\end{definition}

Definition~\ref{def:obs} implies that at any time $t\geq0$, $P$ causally reconstructs the state of system~\eqref{eq:sys} using (the past values of) some nominal disturbance $\bar{u}$, some measured signal $\bar{y}$, the parameter~$d$, and some initial estimate $\bar{\chi}$.
Considering $\bar{y}\neq y$ provides an additional degree of robustness and accounts for the case where the output model $h$ in \eqref{eq:sys_2} is not exact, e.g., when the data are first transformed or traverse additional networks not captured by $h$, cf. \cite{Knuefer2020} for a more detailed discussion.
Note that for the classical case with $\bar{u}\equiv0$ and $\bar{y} = y$, the estimate~\eqref{eq:RGAS} reduces to
$\beta(|x(t)-\hat{x}(t)|) \leq \beta_x(|\chi-\bar{\chi}|)\eta^t + \int_0^t \eta^{t-\tau}\beta_u(|u(\tau)|)d\tau$.
The integral term in this bound can be viewed as the energy of the true disturbance signal $u$ under fading memory and thus has a reasonable physical interpretation, compare also~\cite{Sontag1998,Praly1996}.
Moreover, the discount factor permits a direct derivation of an ``$L^\infty$-to-$L^\infty$'' error bound (cf.,~\cite[Prop.~1]{Schiller2023b}) and thus combines the advantages of classical and integral ISS properties: 
it is applicable for both unbounded disturbances that have small energy and persistent, bounded disturbances with infinite energy, and furthermore, directly implies that $|x(t)-\hat{x}(t)|\rightarrow0$ if $|u(t)| \rightarrow 0$ for $t\rightarrow\infty$.
Note also that the mapping~\eqref{eq:obs_map} covers full-order state observers, but in particular observers that do not admit a convenient state-space representation, such as moving horizon and full information estimators, cf.~\cite{Schiller2023b} and compare also~\cite{Allan2021} for a similar discussion in a discrete-time setting.
The following proposition proves necessity of i-iIOSS for the existence of an observer mapping~\eqref{eq:obs_map} satisfying~\eqref{eq:RGAS}.

\begin{proposition}
	The system~\eqref{eq:sys} admits an RGAS observer mapping in the sense of Definition~\ref{def:obs} only if it is i-iIOSS.
\end{proposition}
\begin{proof}
	This proof follows similar lines as in \cite[Prop.~2.6]{Allan2021}, \cite[Prop.~3]{Knuefer2020}.
	Consider $\chi_1,{\chi}_2\in\mathcal{X}$, $u_1,{u}_2\in\mathcal{M}_{\mathcal{U}}$, and $d\in\mathcal{M}_{\mathcal{D}}$ yielding $x_i(t) = x(t,\chi_i,u_i,d)$ and $y_i(t) = y(t,\chi_i,u_i,d)$, $i=1,2$ for all $t\geq0$.
	Suppose that the observer $P$~\eqref{eq:obs} is designed to reconstruct the trajectory ${x}_2$ using $\bar{\chi} = {\chi}_2$, $\bar{u}={u}_2$, $\bar{y}={y}_2$.
	By application of~\eqref{eq:RGAS}, it follows that $\alpha_x(|x_2(t)-\hat{x}(t)|)=0$ for all $t\geq0$.
	Now assume that this certain design of $P$ is used to reconstruct the trajectory $x_1$.
	Then, since $\hat{x}(t) = x_2(t)$ for all $t\geq0$, the estimate~\eqref{eq:RGAS} directly yields~\eqref{eq:dIOSS} with $\alpha=\beta$, $\alpha_x=\beta_x$, $\alpha_u=\beta_u$, $\alpha_y=\beta_y$, and because $\chi_1,{\chi}_2,u_1,{u}_2,d$ were arbitrary, the system~\eqref{eq:sys} is i-iIOSS, which finishes this proof.	
\end{proof}

\input{literature.tex}

\appendix

\subsection{Proof of Proposition~\ref{prop:f}}
\label{sec:app_f}
We first derive a bound on the difference of trajectories on a fixed time interval by adapting the results from \cite{Lipovan2000} and \cite{Bihari1956}.
This will also be crucial in proving Claim~\ref{prop:c} in~Section~\ref{sec:app_c}.

\begin{lemma}\label{lem:f}
	Let Assumption~\ref{ass:f} hold. Then, there exists some $\rho\in\mathcal{K}_{\infty}$ such that for each $\chi_1,{\chi}_2\in\mathcal{X}$, $u_1,{u}_2\in\mathcal{M}_{\mathcal{U}}$, and $d_1,{d}_2\in\mathcal{M}_{\mathcal{D}}$, there exists $T>0$ such that
	\begin{equation}\label{eq:lem_f}
		|x(t,\chi_1,u_1,d_1)-{x}(t,{\chi}_2,{u}_2,d_2)| \leq\ \rho^{-1}(\rho(c)e^t)
	\end{equation}
	for all $t\in[0,T)$ with
	\begin{equation}\label{eq:lem_f_c}
		c := |\chi_1-{\chi}_2| + T{\kappa}_1(3\|u_1-{u}_2\|_{0:T}) + T{\kappa}_1(3\|d_1-{d}_2\|_{0:T}).
	\end{equation}
\end{lemma}

\begin{proof}	
	Consider arbitrary $\chi_1,{\chi}_2\in\mathcal{X}$, $u_1,{u}_2\in\mathcal{M}_{\mathcal{U}}$, $d_1,{d}_2\in\mathcal{M}_{\mathcal{D}}$.
	The existence of the trajectories $x_i(t)=x(t,\chi_i,u_i,d_i)$, $t\in [0,t_i(\chi_i,u_i,d_i))$ is ensured for some $t_i(\chi_i,u_i,d_i)>0$, $i=1,2$ by continuity of $f$ (Assumption~\ref{ass:f}) and Peano's existence theorem (cf., e.g., \cite[Th.~2.1]{Hartman1964}).
	Let $T:=\min_{i\in{1,2}}\{t_i(\chi_i,u_i,d_i)\}$.
	Then, for all $t\in[0,T)$, the trajectories $x_1$ and $x_2$ satisfy
	\begin{align*}
		&\ x_1(t)-{x}_2(t)  = \chi_1-{\chi}_2 \\
		&\hspace{1ex} + \int_{0}^{t}(f(x_1(\tau),u_1(\tau),d_1(\tau))-f({x}_2(\tau),{u}_2(\tau),{d}_2(\tau)))d\tau.
	\end{align*}
	Define $v(t) = |x_1(t)-x_2(t)|$, $u_{\Delta}=u_1-{u}_2$, and $d_{\Delta} = d_1-{d}_2$. By applying \eqref{eq:ass_f_k}, the triangle inequality, and the fact that $\kappa_1$ is positive definite and non-decreasing, we can deduce that	
	\begin{equation}
		v(t)
		\leq v(0) + \int_{0}^{t}(\bar{\kappa}_1(v(s)) + \bar{\kappa}_1(|u_{\Delta}(s)|) + \bar{\kappa}_1(|d_{\Delta}(s)|))ds\label{eq:proof_f_0}
	\end{equation}
	with $\bar{\kappa}_1(s) := \kappa_1(3s)$.
	Note that
	\begin{align}
		&\ \int_{0}^{t}(\bar{\kappa}_1(|u_{\Delta}(s)|)  + \bar{\kappa}_1(|d_{\Delta}(s)|))ds\nonumber\\
		\leq&\ T(\bar{\kappa}_1(\|u_{\Delta}\|_{0:T}) + \bar{\kappa}_1(\|d_{\Delta}\|_{0:T})). \label{eq:proof_f_4}
	\end{align}
	By combining \eqref{eq:proof_f_0}, \eqref{eq:proof_f_4}, and the definition of $c$ from~\eqref{eq:lem_f_c}, we obtain
	\begin{equation}
		v(t)\leq c +  \int_{0}^{t}\bar{\kappa}_1(v(s)) ds. \label{eq:proof_f_1}
	\end{equation}
	We first assume that $c>0$.
	Denote by $U(t)$ the right-hand side of~\eqref{eq:proof_f_1}. Then, $U(0)=c$ and
	\begin{align}
		\dot{U}(t) &= \bar{\kappa}_1(v(t)) \leq \bar{\kappa}_1(U(t)). \label{eq:proof_f_2}
	\end{align}
	Now consider $G(s) := \int_{1}^{s}\frac{dr}{\bar{\kappa}_1(r)}$ for $s>0$.
	By Assumption~\ref{ass:f}, $\lim_{s\rightarrow 0^+}G(s)=-\infty$ and $\lim_{s\rightarrow \infty}G(s)=\infty$.
	Furthermore, from the Leibniz integral rule, it follows that
	\begin{equation}
		\frac{d}{dt}G({U}(t)) = \frac{d}{dt} \int_{1}^{U(t)}\frac{dr}{\bar{\kappa}_1(r)}
		= \frac{\dot{U}(t)}{\bar{\kappa}_1(U(t))}. \label{eq:proof_f_3}
	\end{equation}
	The combination of \eqref{eq:proof_f_2} and \eqref{eq:proof_f_3} yields $\frac{d}{dt}G({U}(t)) \leq 1$.
	An integration on $[0,t]$ leads to
	\begin{equation}
		G({U}(t))- G({U}(0)) \leq t \Leftrightarrow e^{G({U}(t))} \leq e^{G({U}(0))}e^t. \label{eq:proof_f_5}
	\end{equation}
	Now define $\rho(s):=e^{G(s)}$ for all $s>0$ and $\rho(0):=0$. It follows that $\rho\in\mathcal{K}_{\infty}$ (and thus $\rho^{-1}\in\mathcal{K}_{\infty}$).
	Since $v(t)\leq U(t)$ for all $t\in[0,T)$ and $U(0)=c$, from \eqref{eq:proof_f_5} and the definition of $\rho$ we can conclude that $v(t) \leq \rho^{-1}(\rho(c)e^t)$ for all $t\in[0,T)$.

	It remains to show that \eqref{eq:lem_f} also applies for $c=0$. Performing the same steps as before with $\epsilon>0$ instead of $c$ leads to $v(t) \leq \rho^{-1}(\rho(\epsilon)e^t)$. Letting $\epsilon\rightarrow0$ recovers \eqref{eq:lem_f} for $c=0$ and thus concludes this proof.
\end{proof}

\begin{myproof}{Proof of Proposition~\ref{prop:f}}
	Proposition~\ref{prop:f} is an immediate consequence of Lemma~\ref{lem:f}.
	First, we claim that solutions exist globally in time.
	Indeed, suppose not.
	Then, there exist $\chi\in \mathcal{X}$, $u\in\mathcal{M}_{\mathcal{U}}$, $d\in\mathcal{M}_{\mathcal{D}}$, and some finite time $T_1>0$ such that $\lim_{t\rightarrow T_1} |x(t)| = \infty$, where $x(t) = x(t,\chi,u,d)$.
	Applying Lemma~\ref{lem:f} with $\chi_1=\chi$, ${u}_1=u$, $d_1=d$, and $\chi_2=0$, ${u}_2\equiv 0$, ${d}_2\equiv 0$,~\eqref{eq:lem_f} yields
	$|x(t)| \leq \rho^{-1}(\rho(|\chi| + T_1(\bar{\kappa}_1(\|u\|_{0:T_1}))+\bar{\kappa}_1(\|d\|_{0:T_1}))e^t)$
	for $t\in[0,T_1)$.
	The right-hand side is bounded for $t\rightarrow T_1$, which contradicts finite escape time and hence implies that solutions exist globally on $\mathbb{R}_{\geq0}$.
	
	It remains to show uniqueness of solutions. To this end, assume that $x_1(t)=x(t,\chi,u,d)$ and $x_2(t)=x(t,\chi,u,d)$ represent two solutions of \eqref{eq:sys_1} on the interval $[0,T_2]$ for $T_2>0$ with the same initial conditions $\chi\in\mathcal{X}$ and inputs $u\in\mathcal{M}_{\mathcal{U}}$ and $d\in\mathcal{M}_{\mathcal{D}}$.
	It follows that $c=0$ in~\eqref{eq:lem_f_c} and $|x_1(t)-{x}_2(t)|=0$ for all $t\in[0,T_2]$ by~\eqref{eq:lem_f}, which proves uniqueness of solutions on $[0,T_2]$ and hence concludes this proof.
\end{myproof}

\subsection{Proof of Claim~\ref{prop:c}}
\label{sec:app_c}

To prove continuity, we first need an additional lemma.

\begin{lemma}\label{lem:xy_bound}
	Let Assumptions~\ref{ass:f} and~\ref{ass:h} hold. For every $T,r_\chi,r_u>0$, there exist constants $R_x(T,r_\chi,r_u)>0$ and $R(T,r_\chi,r_u)>0$ such that	
	\begin{align*}
		|x(t,\chi_1,u_1,d)-{x}(t,{\chi}_2,{u}_2,d)| &\leq R_x(T,r_\chi,r_u),\\
		|y(t,\chi_1,u_1,d)-{y}(t,{\chi}_2,{u}_2,d)| &\leq R_y(T,r_\chi,r_u)
	\end{align*}
	for all $t\in[0,T]$, all $\chi_1,\chi_2\in\mathcal{X}$ satisfying $|\chi_1-{\chi}_2|\leq r_\chi$, all $u_1,u_2\in\mathcal{M}_{\mathcal{U}}$ satisfying $\int_0^\infty\eta^{-s}\alpha(|u_1(s)-u_2(s)|)ds \leq r_u$ for some $\alpha\in\mathcal{K}_{\infty}$ with $\alpha(s)\geq{\kappa}_1(3s)$ for all $s\geq0$ and $\eta\in[0,1)$, and all $d\in\mathcal{M}_{\mathcal{D}}$.
\end{lemma}

\begin{proof}
	For $i=1,2$, let $x_i(t) = x(t,\chi_i,u_i,d)$ and $y_i(t)={x}(t,\chi_i,u_i,d)$, $t\geq0$, where we note that Proposition~\ref{prop:f} applies due to satisfaction of Assumption~\ref{ass:f}.
	Define $u_{\Delta}:=u_1-u_2$.
	We can invoke the same arguments as in the proof of Lemma~\ref{lem:f}, where \eqref{eq:proof_f_4} can be replaced by
	\begin{equation*}
		\int_0^t{\kappa}_1(3|u_{\Delta}(s)|)ds \leq \int_0^\infty \eta^{-s}\alpha(|u_{\Delta}(s)|)ds \leq r_u,
	\end{equation*}
	exploiting that $\eta^{-s}\geq1$ for all $s\geq0$.
	Consequently, we obtain $c=r_{\chi} + r_u$ in~\eqref{eq:lem_f_c},
	which by~\eqref{eq:lem_f} implies that
	\begin{equation*}
		|x_1(t)-{x}_2(t)| \leq \rho^{-1}(\rho(r_{\chi} + r_u)e^T) =: R_x(T,r_\chi,r_u)
	\end{equation*}
	uniformly for all $t\in[0,T]$.
	For the second part, the application of~\eqref{eq:sys_2} in combination with Assumption~\ref{ass:h} leads to
	\begin{align*}
		&\ |y_1(t)-{y}_2(t)|
		\leq \kappa_2(|(x_1(t),u_1(t)) - ({x}_2(t),{u}_2(t))|)\\
		\leq&\ {\kappa}_2(2|x_1(t)-{x}_2(t)|) + {\kappa}_2(2|u_1(t)-{u}_2(t)|)\\
		\leq&\ {\kappa}_2(2R_x(T,r_\chi,r_u)) + {{\kappa}_2(2r_u)}
		=: R_y(T,r_\chi,r_u)
	\end{align*}
	for all $t\in[0,T]$, which finishes this proof.
\end{proof}

\begin{myproof}{Proof of Claim~\ref{prop:c}}
	The proof uses mostly similar arguments as in \cite[Th.~3.5]{Allan2021}, with variations due to the continuous-time setting and the class of inputs considered (in particular, Lemmas~\ref{lem:f} and \ref{lem:xy_bound}). It consists of two parts. First, we show that choosing $(\chi_1,\chi_2)$ in a compact set implies that the right-hand side of \eqref{eq:def_lyap} is the same when restricting $t$, $(u_1,u_2)$ to suitable sets; then, we use this property to establish continuity of $U$.

	\textit{Part I.} Define $\mathcal{B}(C):=\{(\chi_1,{\chi}_2) \in \mathcal{X}\times\mathcal{X} : 1/C\leq |\chi_1-\chi_2|\leq C\}$ for $C\geq1$ and consider $(\chi_1,{\chi}_2)\in\mathcal{B}(C)$.
	Then, for any $\epsilon>0$, there exist inputs  $u^{\epsilon}_1,u^{\epsilon}_2\in\mathcal{M}_{\mathcal{U}}$, $d^{\epsilon}\in\mathcal{M}_{\mathcal{D}}$, and a time $t^{\epsilon}\geq0$ such that
	\begin{align}
		&\ \alpha(|\chi_{\Delta}|) \leq U(\chi_1,{\chi}_2)\label{eq:proof_cont_1_a}\\
		\leq&\ \epsilon +	\lambda^{-t^{\epsilon}/2}\Big(\alpha(|x(t^{\epsilon},\chi_1,u^{\epsilon}_1,d^{\epsilon})-x(t^{\epsilon},\chi_2,u^{\epsilon}_2,d^{\epsilon})|)\nonumber\\
		&- \int_0^{\infty}\lambda^{t^{\epsilon}-\tau}2\alpha_u(|u^{\epsilon}_1(\tau)-u^{\epsilon}_2(\tau)|)d\tau\nonumber\\
		& - \int_0^{t^{\epsilon}}\lambda^{t^{\epsilon}-\tau} \alpha_y(|y(t^{\epsilon},\chi_1,u^{\epsilon}_1,d^{\epsilon})-y(t^{\epsilon},\chi_2,u^{\epsilon}_2,d^{\epsilon})|)d\tau \Big) \nonumber\\
		\leq&\ \epsilon +  \lambda^{-t^{\epsilon}/2}\Big(\alpha_x(|\chi_{\Delta}|)\lambda^{t^{\epsilon}}\nonumber\\
		&\hspace{11ex}- \int_0^{\infty}\lambda^{t^{\epsilon}-\tau}\alpha_u(|u^{\epsilon}_1(\tau)-{u}^{\epsilon}_2(\tau)|) d\tau \Big), \label{eq:proof_cont_1}
	\end{align}
	where $\chi_{\Delta}=\chi_1-\chi_2$ and the last inequality follows from i-iIOSS~\eqref{eq:dIOSS}. Consequently,
	\begin{equation}
		\alpha(|\chi_{\Delta}|) \leq \epsilon + \lambda^{t^{\epsilon}/2}\alpha_x(|\chi_{\Delta}|). \label{eq:proof_cont_2}
	\end{equation}
	Choose $\epsilon\leq \bar{\epsilon}(C):=\alpha(1/C)/2$ and recall that $1/C\leq|\chi_{\Delta}|\leq C$. Thus, \eqref{eq:proof_cont_2} yields
	$\alpha(1/C)/2 \leq \lambda^{t^{\epsilon}/2}\alpha_x(C)$,
	which leads to
	\begin{align}
		t^{\epsilon} \leq 2\log_{\lambda}\left(\frac{\alpha(1/C)}{2\alpha_x(C)}\right)=: T(C),
	\end{align}
	where $0<\frac{\alpha(1/C)}{2\alpha_x(C)}<1$. From \eqref{eq:proof_cont_1} and the fact that $\epsilon\leq\bar{\epsilon}(C)$, we also obtain
	\begin{align*}
		&\ \int_0^{\infty}\lambda^{t^{\epsilon}-\tau}\alpha_u(|u^{\epsilon}_1(\tau)-{u}_2^{\epsilon}(\tau)|)d\tau\\
		\leq&\ \alpha_x(C)\lambda^{t^{\epsilon}}-  \frac{1}{2}\alpha(1/C)\lambda^{t^{\epsilon}/2} \leq \alpha_x(C).
	\end{align*}
	Since $t^\epsilon\in [0,T(C)]$, it follows that $\lambda^{t^\epsilon}\geq\lambda^{T(C)}$; hence,
	\begin{align*}
		\int_0^{\infty}\lambda^{-\tau}\alpha_u(|u^\epsilon_1(\tau)-{u}^{\epsilon}_2(\tau)|)d\tau
		\leq \alpha_x(C)\lambda^{-T(C)}=: r_u(C).
	\end{align*}
	As a result, we can infer that 
	$(u_1^\epsilon,u_2^\epsilon)\in\mathcal{B}_u(C):=\{(u_1,u_2)\in\mathcal{M}_{\mathcal{U}}\times\mathcal{M}_{\mathcal{U}} : \int_0^{\infty}\lambda^{-\tau}\alpha_u(|u^\epsilon_1(\tau)-{u}^{\epsilon}_2(\tau)|)d\tau \leq r_u(C)\}$.
		
	\textit{Part II}. Now consider some $\tilde{\chi}_1,\tilde{\chi}_2\in\mathcal{X}$ with $\tilde{\chi}_1\neq\tilde{\chi}_2$. Set $C=2\max\{|\tilde{\chi}_1-\tilde{\chi}_2|,1/|\tilde{\chi}_1-\tilde{\chi}_2|\}\geq1$.
	From the first part of this proof, for $({\chi}_1,{\chi}_2)\in\mathcal{B}(C)$, there exist $\epsilon\in(0,\bar{\epsilon}(C)]$, $(u^{\epsilon}_1,u^{\epsilon}_2)\in\mathcal{B}_u(C)$, $d^{\epsilon}\in\mathcal{M}_{\mathcal{D}}$, and $t^{\epsilon}\in[0,T(C)]$ such that \eqref{eq:proof_cont_1} holds.
	Define
	\begin{align*}
		&x_1(t) := x(t,\chi_1,u^{\epsilon}_1,d^{\epsilon}),
		&x_2(t) := x(t,\chi_2,u^{\epsilon}_2,d^{\epsilon}),\\
		&\tilde{x}_1(t) := x(t,\tilde{\chi}_1,u^{\epsilon}_1,d^{\epsilon}),
		&\tilde{x}_2(t) := x(t,\tilde{\chi}_2,u^{\epsilon}_2,d^{\epsilon}),\\
		&y_1(t) := y(t,\chi_1,u^{\epsilon}_1,d^{\epsilon}),
		&y_2(t) := y(t,\chi_2,u^{\epsilon}_2,d^{\epsilon}),\\
		&\tilde{y}_1(t) := y(t,\tilde{\chi}_1,u^{\epsilon}_1,d^{\epsilon}),
		&\tilde{y}_2(t) := y(t,\tilde{\chi}_2,u^{\epsilon}_2,d^{\epsilon})\phantom{,}
	\end{align*}
	for all $t\in[0,T(C)]$.
	The trajectories $\tilde{x}_1(t)$ and $\tilde{x}_2(t)$ satisfy
	\begin{align}
	U(\tilde{\chi}_1,\tilde{\chi}_2) \geq &\ \lambda^{-t^{\epsilon}/2}\Big( \alpha(|\tilde{x}_1(t^{\epsilon})-\tilde{x}_2(t^{\epsilon})|)  \nonumber\\
	&\hspace{3ex} - \int_0^{\infty}\lambda^{t^{\epsilon}-\tau}2\alpha_u(|u^\epsilon_1(\tau)-{u}^{\epsilon}_2(\tau)|)d\tau\nonumber\\
	&\hspace{3ex} - \int_0^{t^{\epsilon}}\lambda^{t^{\epsilon}-\tau} \alpha_y(|\tilde{y}_1(\tau)-\tilde{y}_2(\tau)|)d\tau\Big).\label{eq:proof_cont_1_b}
	\end{align}
	The combination of \eqref{eq:proof_cont_1_a} and \eqref{eq:proof_cont_1_b} yields
	\begin{align}
	&\hspace{3.4ex} U(\chi_1,{\chi}_2)-U(\tilde{\chi}_1,\tilde{\chi}_2)\nonumber\\
	&\leq\ \epsilon
	+ \lambda^{-t^{\epsilon}/2}\Big(
	\alpha(|x_1(t^{\epsilon})-x_2(t^{\epsilon})|)- \alpha(|\tilde{x}_1(t^{\epsilon})-\tilde{x}_2(t^{\epsilon})|)\nonumber\\
	& \ + \int_0^{t^{\epsilon}}\hspace{-0.5ex}\lambda^{t^{\epsilon}{-}\tau}(\alpha_y(|\tilde{y}_1(\tau){-}\tilde{y}_2(\tau)|){\,-\,}\alpha_y(|{y}_1(\tau){-}{y}_2(\tau)|) 
	)d\tau
	\Big).\label{eq:proof_cont_A}
	\end{align}
	
	Without loss of generality, we assume that\footnote{
		If this is violated, simply replace $\alpha_u$ in~\eqref{eq:dIOSS} by a suitable $\bar{\alpha}_u\in\mathcal{K}_{\infty}$ that majorizes both $\alpha_u$ and ${\kappa_1}(3s)$.
	} $\alpha_u(s)\geq{\kappa_1}(3s)$ for all $s\geq0$.
	By Lemma~\ref{lem:xy_bound}, there exist $R_x,R_y>0$ such that
	\begin{align*}
		\max\{|x_1(t){-}x_2(t)|,|\tilde{x}_1(t){-}\tilde{x}_2(t)|\}
		&\leq R_x(T(C),C,r_u(C))\\
		&=: R_x^C,\\
		\max\{|y_1(t){-}y_2(t)|,|\tilde{y}_1(t){-}\tilde{y}_2(t)|\} &\leq R_y(T(C),C,r_u(C))\\
		& =: R_y^C
	\end{align*}
	uniformly for all $t\in[0,T(C)]$.
	Recall that $\alpha,\alpha_y$ in~\eqref{eq:proof_cont_A} are continuous; hence, they are uniformly continuous on the compact sets $[0,R_x^C]$ and $[0,R_y^C]$, respectively.
	From \cite[Prop.~20]{Allan2017}, there exist $\hat{\alpha},\hat{{\alpha}}_y\in\mathcal{K}_{\infty}$ such that
	\begin{align}
		|\alpha(s_1)-\alpha(s_2)| \leq \hat{\alpha}(|s_1-s_2|), \ s_1,s_2\in [0,R_x^C],\label{eq:proof_cont_hat_a_1}\\
		|\alpha_y(s_1)-\alpha_y(s_2)| \leq \hat{\alpha}_y(|s_1-s_2|), \ s_1,s_2\in [0,R_y^C].\label{eq:proof_cont_hat_a_2}
	\end{align}
	Evaluating the absolute value of the right-hand side of \eqref{eq:proof_cont_A}, using the triangle inequality, applying \eqref{eq:proof_cont_hat_a_1} and \eqref{eq:proof_cont_hat_a_2} followed by the reverse triangle inequality and then the standard one lead us to	
	\begin{align}
		&\ U(\chi_1,{\chi}_2)-U(\tilde{\chi}_1,\tilde{\chi}_2)\nonumber\\
		\leq&\ \epsilon
		+ \lambda^{-t^{\epsilon}/2}\Big(
		\hat{\alpha}\big(|x_1(t^{\epsilon})-\tilde{x}_1(t^{\epsilon})|+ |x_2(t)-\tilde{x}_2(t^{\epsilon})|\big)\nonumber\\
		& + \int_0^{t^{\epsilon}}\lambda^{t^{\epsilon}-\tau}\big(
		\hat{\alpha}_y\big(|{y}_1(\tau)-\tilde{y}_1(\tau)|+|{y}_2(\tau)-\tilde{y}_2(\tau)|\big)
		\big)d\tau
		\Big).\label{eq:proof_cont_3}
	\end{align}
	By applying Lemma~\ref{lem:f} and similar steps as in the proof of Lemma~\ref{lem:xy_bound}, it follows that
	\begin{align}
		|x_i(t)-\tilde{x}_i(t)| &\leq 
		\rho^{-1}(\rho(|\chi_i-\tilde{\chi}_i|)e^{T(C)}), i=1,2,\label{eq:proof_cont_4a}\\
		|{y}_i(t)-\tilde{y}_i(t)| &\leq \kappa_2(\rho^{-1}(\rho(|\chi_i-\tilde{\chi}_i|)e^{T(C)})), i=1,2\label{eq:proof_cont_4b}
	\end{align}
	for all $t\in[0,T(C)]$.
	Hence, from~\eqref{eq:proof_cont_3}, using that $\alpha(|a+b|)\leq \alpha(2a) + \alpha(2b)$ for any $\alpha\in\mathcal{K}$ and $a,b\geq0$ in conjunction with the bounds from \eqref{eq:proof_cont_4a} and \eqref{eq:proof_cont_4b} and the facts that $\lambda^{-t^{\epsilon}/2}\leq\lambda^{-T(C)/2}$ and $\int_{0}^{t^{\epsilon}}\lambda^{t^\epsilon-\tau}d\tau \leq -1/\ln\lambda$, we can infer that there exist $\gamma_x,\gamma_y\in\mathcal{K}$ satisfying
	\begin{align*}
		U(\chi_1,{\chi}_2)-U(\tilde{\chi}_1,\tilde{\chi}_2)
		\leq \epsilon &+ \gamma_x(|\chi_1-\tilde{\chi}_1|) + \gamma_x(|\chi_2-\tilde{\chi}_2|) \\
		& + \gamma_y(|\chi_1-\tilde{\chi}_1|) + \gamma_y(|\chi_2-\tilde{\chi}_2|)\\
		\leq \epsilon &+ \gamma(|\chi_1-\tilde{\chi}_1|) + \gamma(|\chi_2-\tilde{\chi}_2|)
	\end{align*}
	where $\gamma(s):= \gamma_x(s)+\gamma_y(s)$ for all $s\geq0$.	
	Letting $\epsilon\rightarrow0$ and applying a symmetric argument (recall that $(\tilde{\chi}_1,\tilde{\chi}_2)\in\mathcal{B}(C)$ by the definition of $C$) lets us conclude that
	\begin{equation}
		|U(\chi_1,{\chi}_2)-U(\tilde{\chi}_1,\tilde{\chi}_2)| \leq \gamma(|\chi_1-\tilde{\chi_1}|) + \gamma(|\chi_2-\tilde{\chi}_2|).\label{eq:proof_cont_4}
	\end{equation}
	Since $\mathcal{B}(C)$ contains all pairs $(x_1,x_2)$ within a neighborhood of $(\tilde{\chi}_1,\tilde{\chi}_2)$, \eqref{eq:proof_cont_4} implies that $U$ is continuous at each $(\tilde{\chi}_1,\tilde{\chi}_2)\in\mathcal{X}\times\mathcal{X}$ for $\tilde{\chi}_1\neq\tilde{\chi}_2$.
	It remains to show that $U$ is also continuous at $(\tilde{\chi},\tilde{\chi})$.
	To this end, consider any $({\chi}_1,{\chi}_2)\in\mathcal{X}\times\mathcal{X}$; since $U(\tilde{\chi},\tilde{\chi})=0$, it follows that
	\begin{align*}
		&\ |U({\chi}_1,{\chi}_2){\,-\,}U(\tilde{\chi},\tilde{\chi})| = U({\chi}_1,{\chi}_2) 
		\leq \alpha_x(|{\chi}_1{\,-\,}{\chi}_2|)\\
		\leq&\ \alpha_x(|{\chi}_1{\,-\,}\tilde{\chi}|{\,+\,}|\tilde{\chi}{\,-\,}{\chi}_2|) \leq \alpha_x(2|{\chi}_1{\,-\,}\tilde{\chi}|){\,+\,}\alpha_x(2|\chi_2{\,-\,}\tilde{\chi}|),
	\end{align*}
	which implies that $U$ is continuous at $(\tilde{\chi},\tilde{\chi})$. Hence, $U$ is continuous on $\mathcal{X}\times\mathcal{X}$, which finishes this proof.	
\end{myproof}

\begin{remark}\label{rem:local_L}
	Part I of the proof of Claim~\ref{prop:c} gives rise to the fact that $(u_1^{\epsilon},u_2^{\epsilon})\in\mathcal{B}_u(C)$, i.e., the inputs $u_1^{\epsilon}$ and $u_2^{\epsilon}$ are such that the weighted ``energy'' of its difference is located in a ball of radius $r_u$ centered at the origin.
	However, this implies no information about the absolute range of $u_1$ and $u_2$, which prevents the use of a local Lipschitz property of $f$ to bound the evolution of the difference of state trajectories in the proof of Claim~\ref{prop:c} below~$\eqref{eq:proof_cont_A}$ and in \eqref{eq:proof_cont_4a}.
	In contrast, the global nature of Assumption~\ref{ass:f} allows the derivation of such a bound, and the conditions in~\eqref{eq:ass_f} ensure that it is finite for any finite $t$.	
\end{remark}

\end{document}